\documentclass{ipi}

\usepackage{amsmath}
\usepackage{amssymb}
\usepackage{paralist}
\usepackage{caption,graphicx,subcaption,mathtools}

\def\currentvolume{X}
 \def\currentissue{X}
  \def\currentyear{20xx}
   
    \def\ppages{X--XX}
     \def\DOI{10.3934/ipi.xx.xx.xx}
\newtheorem{theorem}{Theorem}[section]

\newtheorem{lemma}[theorem]{Lemma}

\theoremstyle{definition}

\newtheorem{remark}{Remark}

\def\be{\begin{equation}}
\def\ee{\end{equation}}
\def\bea{\begin{eqnarray}}
\def\eea{\end{eqnarray}}

\newcommand{\R}{{\mathbb R}}

\author[Hari Nortunen and Mikko Kaasalainen]{}
\email{hari.nortunen@tut.fi}
\email{mikko.kaasalainen@tut.fi}
\title[Distributions from random projections]{Shape and spin distributions of large object populations from random projection areas}
\begin{document}

\maketitle
\medskip
\centerline{\scshape Hari Nortunen}
\medskip
\medskip
\centerline{\scshape Mikko Kaasalainen}
\medskip
{\footnotesize
 \centerline{Department of Mathematics}
   \centerline{Tampere University of Technology}
   \centerline{PO Box 553, 33101 Tampere, Finland}
}
\bigskip
\begin{abstract} 
We model the shape and spin characteristics of an object population when there are not enough data to model its single members. The data are random projection areas of the members. We construct a mapping $f(x)\rightarrow C(y)$, $x\in\R^2$, $y\in\R$, where $f(x)$ is the distribution function of the shape elongation and spin vector obliquity, and $C(y)$ is the cumulative distribution function of an observable $y$ describing the variation of the observed projection areas of one member, and show that the mapping is invertible. Using the projected area of an ellipsoid as our model, we obtain analytical basis functions for a function series of $C(y)$ and prove uniqueness and stability properties of the inverse problem. Even though the model error is considerably larger than the measurement noise for realistic cases of arbitrary shapes (such as asteroids), the main characteristics of $f(x)$ (such as the locations of peaks) are robustly recovered from the data.
\end{abstract}
\section{Introduction}

The goal of this paper is to model the characteristics of a population when there are not enough data to 
model its single members. Our case study is the distribution of the shape and spin characteristics of a large number
of rotating objects in $\R^3$ when we have random observations of the areas of their projections in different viewing geometries. This setup corresponds to the sporadic observations of the brightnesses of the asteroids in our solar system.

In the following, we consider some choices of observables and the 
corresponding models, cumulative distribution functions (CDF), and
distribution functions (DF).
With DFs, we essentially take each observation to be an identical procedure,
a repeated sample of the
distribution. The targets lose their identity: 
observations of the same target at various times
can be taken as independent samples of the DF.  Our setup belongs to the general class of problems of the relation between some model DF $f(x)$, $x\in\R^n$, and the corresponding distribution $g(y)$ in some space of observables $y\in\R^k$, when the data are samples of $g$. In the multidimensional case, this can be solved with likelihood methods or likelihood-free inference (e.g., \cite{dyn}), but if $y\in\R$ as here, the case is easier as the samples can be examined by forming a single CDF.

We use CDFs of observables as they are well-defined non-binned
directly measurable distribution quantities. Thus we proceed by computing
the model CDF from the model DFs and comparing it with the data CDF.
Our model CDF is determined by calculating by integration how many model configurations
can contribute to the observed CDF at each value of the observable. The main problem is
whether the thus obtained mapping
$$
f(x)\rightarrow C(y),\quad f\ge 0,\,x\in\R^n,\,0\le C\le a,\,y\in\R,
$$
is invertible. Here $f$ is the DF of some intrinsic properties $x$, and $C$ is the CDF of some observable $y$.

The physical realization of our case study was originally introduced by Szab\'o et al. \cite{Szabo}. Their study did not contain any analytical or numerical inspection of the generic inverse problem, especially its uniqueness and stability properties. Their examination included over $10^4$ pieces of actual asteroid data, and they concluded that while the data are insufficient for obtaining the properties of individual bodies, a statistical analysis is possible. The asteroid population was treated as a distribution function. The approach of \cite{Szabo} was innovative, and we aim to expand upon their study by investigating the inverse problem mathematically, including the role of the insufficient model and other assumptions that do not necessarily hold in practice (such as the distribution of spins that was fixed in \cite{Szabo}). 

We present both a theoretical analysis and numerical simulations. We generate cumulative distribution functions of large asteroid population, aiming to study what the CDF reveals about the properties of the population. To keep our model simple and solvable, we choose to utilize as few parameters as possible. First we consider the shape elongation $p$ as our only parameter, while the spin latitude $\beta$ remains uniformly distributed. Then we move on to a more advanced case where the distributions of both $p$ and $\beta$ are to be solved. We show that  it is possible to obtain information about the $\beta$ distribution, which was missing from \cite{Szabo}. 

\section{Observables and forward problem}

Our model shape is the triaxial ellipsoid, since it has a particularly simple analytical expression for the area of its projection in any given viewing direction \cite{con}. In this paper, we use the terms brightness and projection area interchangeably, because they are physically almost the same (up to a scaling factor) for dark targets when the viewing and illumination directions coincide \cite{genproj}. We further simplify the model
(semiaxes $a,b,c$) with $b=c=1$.  Naturally this would be a coarse
shape approximation for individual targets of general shape, but even if our model is actually not very realistic in practice,
it should portray some coarse-scale population tendencies right when we have many
observations. Thus it suffices to have a model that represents
the effects of shape elongation and spin direction in a roughly correct manner.

\subsection{Amplitudes $A$}

First we consider the theory that would hold if our analytical shape model were correct. Then, from Sect.\ \ref{sec:asteroid} on, we discuss the consequences caused by the incorrect model by numerically computing the brightnesses of general shapes, obtaining the brightness amplitude over one rotation of the body, and its cumulative distribution function.

Let us first assume isotropic spins and two model parameters: $p:=b/a$
describing the shape elongation (the smaller the $p$, the more elongated the body), and $\theta$ for aspect angle: 
$\cos\theta={\mathbf v}\cdot{\mathbf e}$, where $\mathbf v$ is the spin direction (given by the polar coordinates $(\beta,\lambda)$ in the inertial frame)
and $\mathbf e$ the line of sight (unit vectors). Due to model symmetry, we only
need to consider the interval $0\le\theta\le\pi/2$. With $\phi$ for the longitudinal angle in a coordinate frame fixed to the ellipsoid, the area $I$ of the ellipsoid's projection in the direction $\mathbf e$ is \cite{con}
$$
I=\pi abc\sqrt{\frac{\sin^2\theta\cos^2\phi}{a^2}+\frac{\sin^2\theta\sin^2\phi}{b^2}+\frac{\cos^2\theta}{c^2}},
$$
so in terms of our model definitions, the brightness $L$ scaled against the maximal possible value $\pi a$ is
\begin{equation}
L=\sqrt{p^2\sin^2\theta\cos^2\phi+\sin^2\theta\sin^2\phi+\cos^2\theta}=\sqrt{1+(p^2-1)\sin^2\theta\cos^2\phi}.\label{bright}
\end{equation}

The statistical observable can be anything that describes the variation of the brightness as the target rotates (at a fixed $\theta$). At first we take this to be the peak-to-peak amplitude; here we consider the ratio $A=L_{\rm min}/L_{\rm max}=L\vert_{\phi=0}/L\vert_{\phi=\pi/2}$ (i.e., an "inverse amplitude": the smaller the $A$,
the larger the variation). Thus we have chosen the convenient $0< p\le 1$ and $0<A\le 1$ (rather than either of these extending to infinity). The
assumption is that all objects rotate about an axis (the ellipsoid's $c$-axis), which produces the observed projections random in $\phi$. At first, we consider the randomness of $\theta$ to be due to a uniform distribution of rotation axis directions on $S^2$; later, we take the randomness to be caused by a shifting viewing position.

The
amplitude $A$ is given by
\begin{equation}
A=\sqrt{\cos^2\theta+p^2\sin^2\theta}=\sqrt{1+(p^2-1)\sin^2\theta},\label{amp}
\end{equation}
so the iso-$A$ curves in the $(p,\theta)$-plane are given by
$$
\cos^2\theta_A(p)=\frac{A^2-p^2}{1-p^2}:=g_A(p).
$$
The solutions for $\theta_A$ are convex "ripples" starting from the point $(p=0,\theta=\pi/2)$
(upper left corner) for $A=0$ and continuing to the 
lines $\theta=0$ and $p=1$ for $A=1$ (lower right corner).
Denoting the model DF of elongation by $f(p)$, we write the unnormalized CDF $C(A)$ as 
$$
C(A)=\int_0^{p_{\rm max}(A)} f(p)\int_{\theta_A(p)}^{\pi/2} \sin\theta\,
d\theta \,dp,
$$
where the minimal shape elongation needed to produce
amplitude $A$, obtained at $\theta=\pi/2$, is
$p_{\rm max}(A)=A$.
Thus
\begin{equation}
C(A)=\int_0^A f(p)\int_0^{\sqrt{g_A(p)}} \,dx \,dp
=\int_0^A f(p) \sqrt{g_A(p)} \,dp.\label{CAp}
\end{equation}

We can also include the effect of spin distribution. Assuming $\lambda$ to
be isotropic and the observation directions to be in the $xy$-plane of the inertial frame (as they approximately are for the majority of asteroids, when this plane is that of the Earth's orbit), we study the DF $f_\beta(\beta)$ 
(or its joint DF $f(p,\beta)$ with $p$), or the more useful
$f_\beta(\cos\beta)$. The minimal aspect angle is
 $\theta_{\rm min}=\pi/2-\beta$. Now, substituting $ {\bf e}=(\cos\lambda_e,\sin\lambda_e,0)$ into
$\cos\theta=e_1\sin\beta\cos\lambda+e_2\sin\beta\sin\lambda+e_3\cos\beta$,
we have
$$
\cos\theta=\sin\beta\cos\Lambda,
$$
where $\Lambda:=\lambda-\lambda_e$ is assumed isotropic (evenly distributed 
longitudes of 
spins and observing directions). It is sufficient to explore the region
$\Lambda\in[0,\pi/2]$ as other quadrants are just symmetric multiples.

The iso-$\theta$ curves 
$$
\Lambda_\theta (\beta)=\arccos\frac{\cos\theta}{\sin\beta}
$$
in the $(\beta,\Lambda)$-plane are now expanding
``ripples'' of increasing $\theta$
starting from the point $(\beta=\pi/2,\Lambda=0)$ for $\theta=0$.
The CDF for $\theta$ is
\[
\begin{split}
C_\theta(\theta)&=\int_{\pi/2-\theta}^{\pi/2} 
f_\beta(\cos\beta) \sin\beta\int_0^{\Lambda_\theta (\beta)} 
\,d\Lambda' \,d\beta
=\int_{\pi/2-\theta}^{\pi/2} f_\beta(\cos\beta) \sin\beta
\Lambda_\theta (\beta)\,d\beta\\
&=\int_{0}^{\sin\theta}
f_\beta(x)\arccos\frac{\cos\theta}{\sqrt{1-x^2}}\,dx.
\end{split}
\]
(Differentiating
$d C_\theta(\theta)/d\theta$ yields $\sin\theta$ when $f_\beta=1$ as expected for isotropic spins.)

Using the complement of $C_\theta$ (i.e., $\hat C_\theta$ in the decreasing direction from $\theta=\pi/2$ to $\theta=0$) to
write the number of  states between $\theta_A(p)$ and $\theta=\pi/2$, our CDF $C(A)$ is
\begin{equation}
C(A)=\int_0^A \Big\lbrack\frac{\pi}{2}\int_0^1 f(p,x)\,dx -\int_{0}^{\sqrt{1-g_A(p)}}
f(p,x)\arccos\frac{\sqrt{g_A(p)}}{\sqrt{1-x^2}}\,dx\Big\rbrack \,dp.\label{CApx}
\end{equation}

\subsection{Brightness deviation $\eta$ for amplitude estimation} \label{sec:eta}

If the amplitude cannot be measured directly, a possible observable is the brightness variation around some
mean value, requiring fewer points. Using intensity squared, 
we obtain from Eq.\ \eqref{bright} a simple average quantity over model rotation
at constant $\theta$:
$$
\langle L^2\rangle
=\frac{1}{2\pi}\int_0^{2\pi} \Big( 1+\sin^2\theta(p^2-1)\cos^2\phi \Big) \,d\phi=
1+\frac{1}{2}\sin^2\theta(p^2-1).
$$
The standard deviation over rotation is
\[
\begin{split}
\Delta L^2&=\sqrt{\langle(L^2-\langle L^2\rangle)^2\rangle}=
\sqrt{\langle[\sin^2\theta(p^2-1)(\cos^2\phi-1/2
)]^2\rangle}\\
&=\sin^2\theta(1-p^2) \left[ \frac{1}{2\pi}\int_0^{2\pi}(\cos^4\phi-\cos^2\phi)\,d\phi
+\frac{1}{4} \right]^{1/2}=\sin^2\theta(1-p^2)/\sqrt{8},
\end{split}
\]
and normalizing this with $\langle L^2\rangle$ yields
$$
\eta(\theta,p):=\Delta L^2/\langle L^2\rangle
=\sqrt{ \Big\langle \Big( \frac{L^2}{\langle L^2\rangle}-1 \Big)^2 \Big\rangle }
=\frac{1}{2\sqrt{2}}\Big[\frac{1}{\sin^2\theta(1-p^2)}-\frac{1}{2}\Big]^{-1}.
$$
Note that $0\le\eta\le1/\sqrt{2}$.
Thus, by Eq.\ (\ref{amp}), our brightness deviation $\eta$ is directly related to the amplitude $A$:
\begin{equation}
\eta=\frac{1}{\sqrt{8}}\Big(\frac{1}{1-A^2}-\frac{1}{2}\Big)^{-1},
\quad A=\sqrt{1-\Big(\frac{1}{\sqrt{8}\eta}+\frac{1}{2}\Big)^{-1}}.\label{eta}
\end{equation}
This is a particular advantage of the biaxial model: we can use all observations of $A$ and $\eta$ together to form a $C(A)$ instead of having to compute a $C_\eta(\eta)$ by forming a $g_\eta(p)$ and continuing as for $g_A(p)$ above (resulting in similar types of integrals). The latter would be the case for the triaxial ellipsoid, since $\eta$ would depend on $\theta$ (and $c$) in addition to $A$.

\subsection{Brightness two-point scatter $q$} \label{sec:2point}

If there are not enough points covering the rotational phase
to estimate $A$ even by $\eta$, we can use simple
two-point brightness differences. For a group of $N$ points for one target, 
we have $N(N-1)/2$ such values, ordered such that the difference $0<q\le 1$ is
$q=L_{\rm dimmer}/L_{\rm brighter}$. We do not need to have more than one such pair for
one target, so one object does not have to cover
the rotational phases well. This is the observable used in \cite{Szabo}.

As before, we consider the case when $\theta$ is (approximately) the same for the pair. Now we have,
for two rotation phases $\phi_0$ and $\phi$,
$$
\frac{1+(p^2-1)\sin^2\theta\cos^2\phi}
{1+(p^2-1)\sin^2\theta\cos^2\phi_0}=q^2,
$$
so, with $0< q\le 1$, i.e., $\phi\le\phi_0$ (due to symmetry, we only
need to consider the interval $0\le\phi\le\pi/2$), we define iso-$q$ contours
in the $(\phi,\theta)$ plane (for given $p,\phi_0$) by
$$
r(q,p,\phi_0,\phi):=\frac{q^2-1}{(p^2-1)(\cos^2\phi-
q^2\cos^2\phi_0)},
$$
so, to have viable solutions for $\theta_q$ from $\sin^2\theta_q=r$, we must have $p\le q$, $\phi\le\phi_0$, and 
$$
\cos^2\phi\ge \frac{q^2-1}{p^2-1}+q^2\cos^2\phi_0:=s(q,p,\phi_0)\ge\cos^2\phi_0,
$$
so $\phi$ exist for given $p,q,\phi_0$ only if $s\le 1$; i.e.,
$$
\cos^2\phi_0\le \frac{p^2-q^2}{q^2(p^2-1)}:=t(q,p).
$$
Denoting
$$
\tilde s(q,p,\phi_0):=\arccos\sqrt{s(q,p,\phi_0)},\quad \tilde t(q,p):=\arccos\sqrt{t(q,p)},
$$
 our CDF is thus
\begin{equation}
\begin{split}
C_q(q)&=\int_0^q f(p)\int_{\tilde t(q,p)}^{\pi/2} \int_0^{\tilde s(q,p,\phi_0)}
\int_{\theta(q,p,\phi_0,\phi)}^{\pi/2}\sin\theta' 
\,d\theta' \,d\phi \,d\phi_0 \,dp\\
&=\int_0^q f(p)\int_{\tilde t(q,p)}^{\pi/2} \int_0^{\tilde s(q,p,\phi_0)}
\sqrt{1-r(q,p,\phi_0,\phi)} \,d\phi \,d\phi_0 \,dp.
\end{split}\label{Cqp}
\end{equation}
Again, we can include the $\beta$-distribution by expanding the integral in the same way as with Eq.\ \eqref{CApx}.

 \section{Inverse problem}

In this section, we consider the properties of the inverse problem version of the forward model above. First we show that the problem of recovering the $p$-distribution from $\eta$-data has a unique solution (for the simplified model), and the ill-posedness is not severe. Then we discuss practical methods of solution before moving to realistic shapes and numerical simulations in the following sections. We also discuss the reason why $\eta$-data are sufficient for recovering $f(p,\beta)$, while the two-point $q$-data only suffice for $f(p)$.

\subsection{Uniqueness and stability}

\begin{theorem}
The distribution function $f(p)$ is uniquely derivable from $C(A)$ (i.e., $\eta$-scatter data), and the problem is (moderately) ill-posed.
\end{theorem}

\begin{proof}
We can always assume that the observed $C(A)$ can be expressed as a polynomial ($C$ is bounded and $0\le A\le 1$) to an arbitrary precision and degree by, e.g., linear sets of equations or orthogonal functions (since this construction is not used in practice, it is immaterial how the expansion is obtained). Then, by the lemma of the Appendix, the polynomial coefficients of $f(p)$ are uniquely determined, and any errors in the coefficients of $C(A)$ blow up at infinite degree at a rate proportional to the degree. The results holds for $\eta$-data as well, since $A$ is uniquely derivable from $\eta$.
\end{proof}

\subsection{Solution methods}

In an approximation consistent with the coarseness of the model, it is practical to assume the population to consist of a moderate number $n$ of bins in each of which all members have the same $p$ (and $\beta$). Then, if we have only
$p$-bins and isotropic $\theta$,
$$
C(A)=\sum_{i =1}^n w_i\,F_i(A),
$$
where the basis functions $F_i(A)$ are, from Eq. (\ref{CAp}),
\begin{equation}
F_i(A)=\left\{\begin{array}{rl}
0, & A\le p_i\\
\sqrt{\frac{A^2-p_i^2}{1-p_i^2}} ,& A>p_i.
\end{array}\right.
\end{equation}
The range of the monotonously increasing $F_i$ is $[0,1]$, and $F_i=1$ at $A=1$ (Fig.\ \ref{Fi}). The occupation numbers of the bins are given
by $w_i$.  In matrix form,
$$
Mw = C , \quad \text{where} \quad C\in\R^k, w\in\R^n,\, M_{ji}=F_i(A_j),
$$
where the $k$ observed values of $A$ are sorted in ascending order, and the vector $C$ contains the observed CDF: each element $C_j = j/k$ is the value of $C$ at $A_j$.

\begin{figure}
\centering \includegraphics[width=\textwidth]{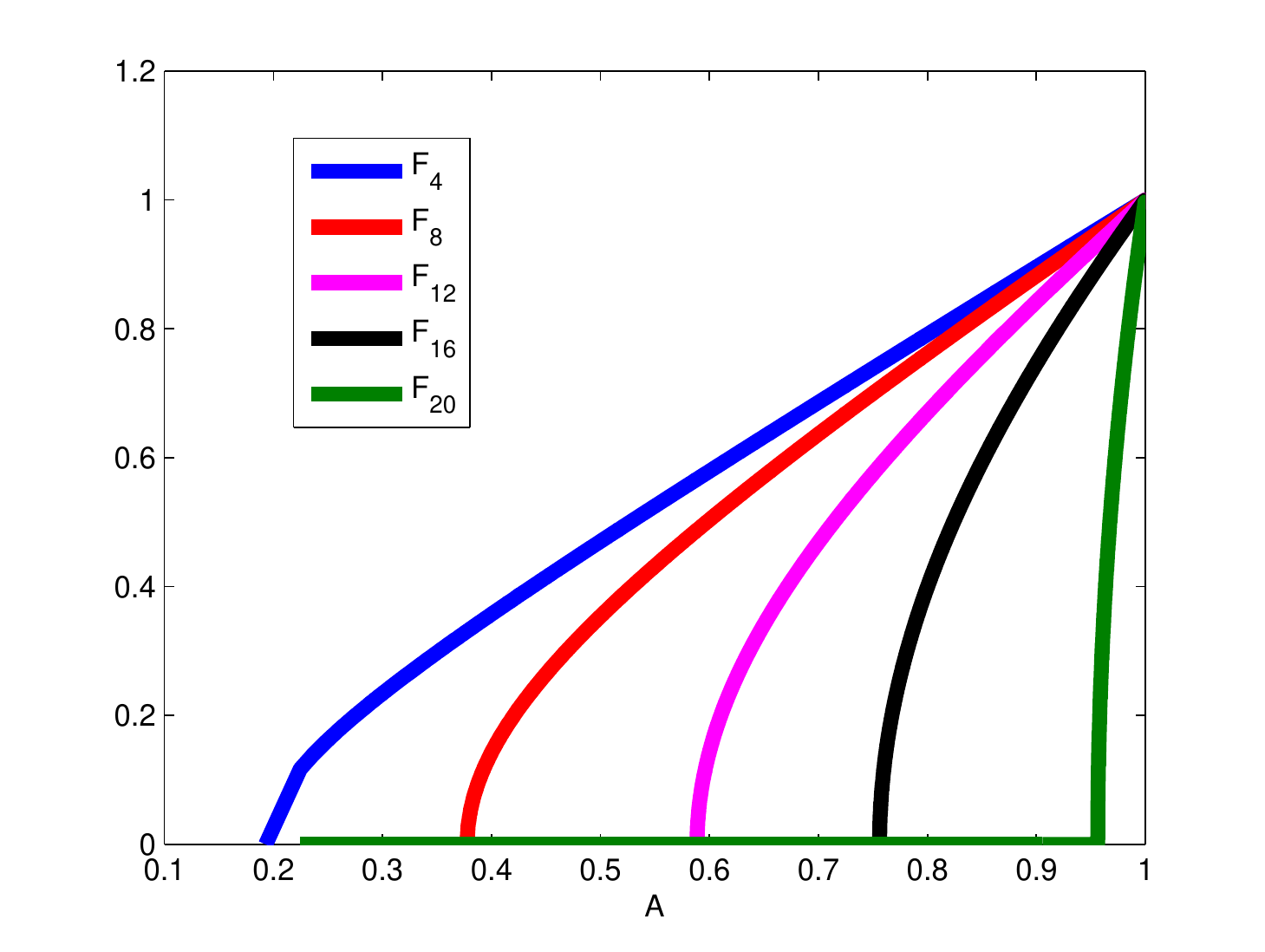}
\caption{Sample basis functions $F_i$ on a set of bins $p_i$, where $i = 1$, $\ldots$, $20$.}
\label{Fi}
\end{figure}

If we include the $\beta$-distribution,
$$
C(A)=\sum_{ij} w_{ij}\,F_{ij}(A),
$$
where, from Eq.\ (\ref{CApx}), the monotonously increasing basis functions $F_{ij}(A)$, with the range $[0,\pi/2]$, are, on a $(p_i,\beta_j)$-grid,
\begin{equation}
F_{ij}(A)=\left\{\begin{array}{rl}
0, & A\le p_i\\
\frac{\pi}{2}-\arccos\frac{\sqrt{A^2-p_i^2}}{\sin\beta_j\sqrt{1-p_i^2}} ,& p_i<A<\sqrt{\sin^2\beta_j+p_i^2\cos^2\beta_j}\\
\frac{\pi}{2}, & A\ge\sqrt{\sin^2\beta_j+p_i^2\cos^2\beta_j}.
\end{array}\right.\label{Fijeq}
\end{equation}
The $F_{ij}(A)$ are sigmoidal functions (Fig.\ \ref{Fij}), approaching the step function when $p_i\rightarrow 1$ (step at $A=1$) or $\beta_j\rightarrow 0$ (step at $A=p_i$). Because of our choice of scale of $p$ and $A$, parts of the $F_{ij}$ tend to pack together at the low end of $A$, making them less well distinguishable than those with the slope in the higher end of $A$, but on the other hand, $p$-values less than 0.4 are not likely for real celestial bodies.
The occupation numbers $w_{ij}$ are assigned to each bin, and
occupation levels proportional to $\sin\beta$ mean a uniform density on the direction sphere. 

\begin{figure}
\centering \includegraphics[width=\textwidth]{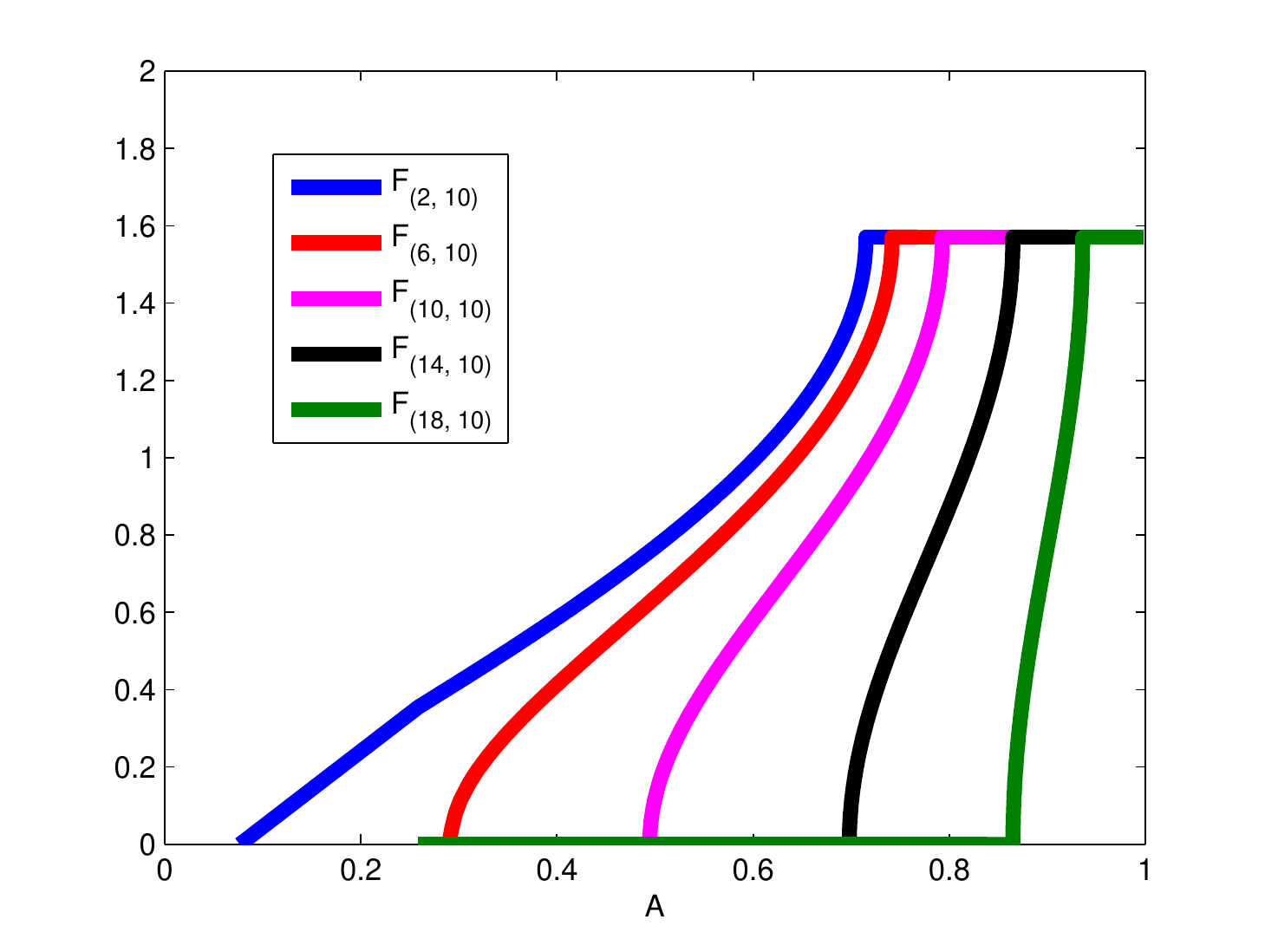}
\caption{Sample basis functions $F_{ij}$ on a set of bins $(p_i , \beta_j)$, where $i = 1$, $\ldots$, $20$ and $j = 1$, $\ldots$, $19$.}
\label{Fij}
\end{figure}

\begin{theorem}
The bin model $w_{ij}$ are uniquely determined by the $C(A)$.
\end{theorem}
\begin{proof}
The pairs of end points; i.e., the values of $A$ between which $F_{ij}$ changes ($A_-$ at $F_{ij}=0$ and $A_+$ at $F_{ij}=\pi/2$), are unique for each $F_{ij}$. Any combination of $F_{ij}$ will start to deviate from zero at the lowest $A_-$ of the set, and stop changing at the highest $A_+$ of the set. Thus both end points of an $F_{ij}$ cannot be matched by a superposition of other $F_{kl}$, so the $F_{ij}$ are linearly independent. Since the model $C(A)$ is a linear combination of the $F_{ij}$, the $w_{ij}$ are unique for the observed $C(A)$.
\end{proof}

The basis function $G_i(q)$ for a given $p_i$ in the two-point brightness scatter case is, from Eq.\ (\ref{Cqp}),
\begin{equation}
G_i(q)=\left\{\begin{array}{rl}
0, & q\le p_i\\
\int_{\tilde t(q,p_i)}^{\pi/2} \int_0^{\tilde s(q,p_i,\phi_0)}
\sqrt{1-r(q,p_i,\phi_0,\phi)} \,d\phi \,d\phi_0 ,& q>p_i.
\end{array}\right.
\end{equation}
Although the $\phi$-integral can be given in terms of elliptic functions, this is best computed by evaluating 
the double integral numerically. The maximum value of $G_i(q)$ is obtained at $q=1$:
$$
G_i(1)=\int_0^{\pi/2}\int_0^{\phi_0}\,d\phi\,d\phi_0=\frac{\pi^2}{8}.
$$
Our basis functions $G_i$ are closed-form expressions of those computed by Monte-Carlo sampling in \cite{Szabo}.

In principle, we can expand $G_i$ to $G_{ij}(q)$ for a $(p_i,\beta_j)$-grid in the same way that $F_i$ were expanded to $F_{ij}$.
However, a notable difference between the two-index basis functions of $A$- or $q$- data is that the $G_{ij}(q)$ all reach their maxima at the same point $q=1$ since the two-point comparison can always contain two equal brightnesses for any $p$ and $\beta$. Thus the $\beta_j$-curves of the $G_{ij}(q)$ of a given $p_i$ form a curve family with the same abscissae for the minimum ($q=p_i$) and maximum ($q=1$); i.e., members of the family can easily be mimicked by a superposition of other members unlike in the case of $F_{ij}(A)$. Thus $G_{ij}$ are not usable for solving the inverse problem in practice; i.e., $\beta$-information is not recoverable from $q$-data as we will note below.

If we want to use regularization to smooth the solutions for either $p$ or $\beta$, we may apply, e.g., the following $(n-1)\times n$ regularization matrix in the $p$-only case:
$$
(R_p)_{ij} = \left\{ \begin{array}{rl}
-1/(p_{i+1}-p_i), & i=j \\
1/(p_{i+1} -p_i), & j=i+1 \\
0, & {\rm elsewhere}
\end{array} \right.
$$
and its generalization for the $(p,\beta)$-grid, as well as similarly $R_\beta$ with $\beta$. These approximate the gradients at each $w_{ij}$ in the $p$- and $\beta$-directions only; one can construct more general matrices, but we found these to suffice for our problem. The occupation numbers can be obtained as a solution to an optimization problem:
\begin{equation}
\hat w = \arg\min_w \left( \Vert{C-Mw}\Vert^2 + \delta_p \Vert{R_p w}\Vert^2 + \delta_{\beta} \Vert{R_{\beta} w}\Vert^2 \right),\quad w\in\R^n_+ .\label{opt}
\end{equation}
To obtain the solution $\hat{w}$, we create an extended matrix $\tilde{M}$:
$$
\tilde M=\left(\begin{array}{r}
M\\
\sqrt\delta_p R_p\\
\sqrt\delta_{\beta} R_{\beta} \end{array}\right), \quad \tilde C=\left(\begin{array}{l}C\\0_{(l-1)m}\\
0_{l(m-1)}\end{array}\right),
$$
assuming a $(p, \beta)$-grid of the size $n=lm$ with, respectively, $l$ and $m$ equally spaced $p$- and $\beta$-values, and we find the least-squares solution of $\tilde{M} w = \tilde{C}$ with the constraint that each element of $w$ be larger than or equal to zero. Due to the instability of the problem, the direct unconstrained matrix solution would lead to negative values, but in, e.g., the Matlab environment, the positivity constraint is simple to enforce with a function that uses quadratic programming. We found that this is more practical than nonlinear optimization with, e.g., $w_i = \exp(z_i)$.

\section{General shapes and model error} \label{sec:asteroid}

The ellipsoid is a very crude approximation of an asteroid, so we test our approach with simulated data created with more general shapes. These can produce brightness variation over $\phi$ that is much more complicated than that of the ellipsoid, even when the shapes are convex \cite{nea}.
For such shapes, the concept of elongation is no longer as well defined as for ellipsoids, but we estimate $p = b/a$ simply by choosing $a$ to be the longest diameter in the $xy$-plane, and $b$ the width in the corresponding orthogonal direction.

For convex shapes, the projected area (i.e., brightness with geometric scattering) can be computed as a sum over visible facets of a polyhedral representation:
$$
L(\phi, \theta) = \sum_{i : \, \mu_i > 0} \mathcal{A}_i \mu_i ,
\label{eq:Lfun-ast}
$$
where $\mu_i = n_i \cdot \omega$, $\mathcal{A}_i$ is the area of the facet, $n_i$ is the outward unit normal vector, and $\omega$ is the direction of the illumination source. This makes simulations very fast (for non-convex shapes, ray tracing is required to determine which facets are visible, but their brightness variation can be well described by convex shapes when viewing and illumination directions are close to each other). 
We utilize asteroid shape models available at DAMIT\footnote{\tt http://astro.troja.mff.cuni.cz/projects/asteroids3D} to create our data. We also generate additional, artificial shapes by applying basic transformations (such as stretching and shrinking) on the original shape models to populate $p$-bins at will. The DAMIT shapes reproduce the typical brightness variations seen in actual asteroids, so with them we can extensively simulate real data. When CDFs of actual observations are created, one should choose data from as closely coinciding viewing and illumination directions as possible to reduce further modelling errors.

In Eq.\ \eqref{eta}, the condition $0 \le \eta \le 1/\sqrt{2}$ may be violated at some measurements of $L(\phi, \theta, p)$ when the parameter $p$ is low ($\lesssim 0.4$). If $\eta > 1/\sqrt{2}$, it follows that the amplitude $A$ becomes purely imaginary according to Eq.\ \eqref{eta}. This error is caused by the model, as the ellipsoidal approximation of a convex asteroid is inaccurate. We have omitted these complex amplitudes in our study. In practice, the majority of realistic values of $p$ are in the range of $[1/2, 1]$, so it is rare that complex amplitudes are encountered.

\begin{figure}[!ht] \centering \includegraphics[width=\textwidth]{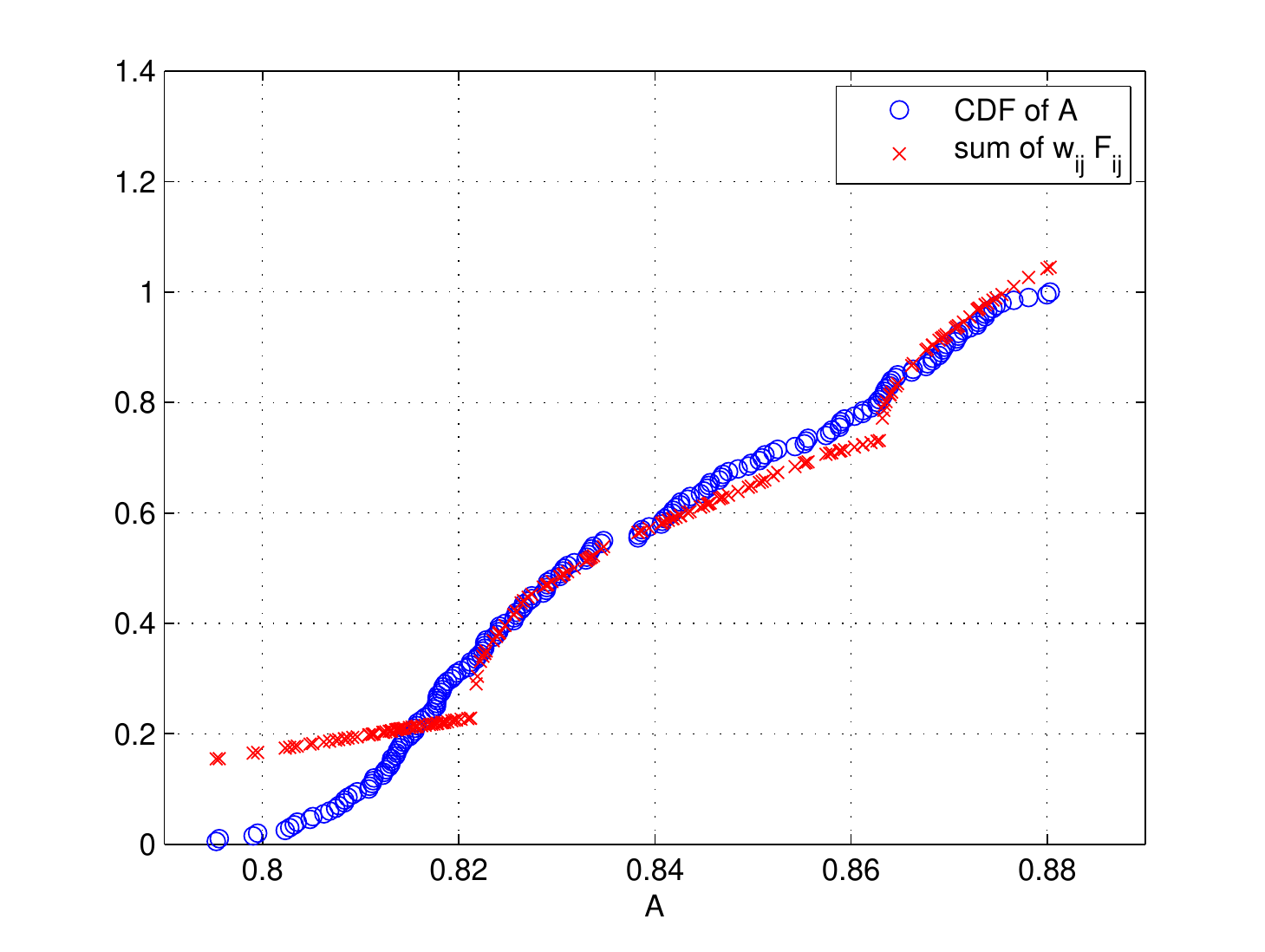} \caption{Recovered function series $\sum_{ij} w_{ij} F_{ij}(A)$ plotted together with the data CDF $C(A)$. Here the occupation numbers $w_{ij}$ were solved as an inverse problem for a single asteroid. This shows that if one target of general shape is used for the CDF simulations, the basis functions $F_{ij}$ are not very useful.} \label{fig:fitting01} \end{figure}
\begin{figure}[!ht] \centering \includegraphics[width=\textwidth]{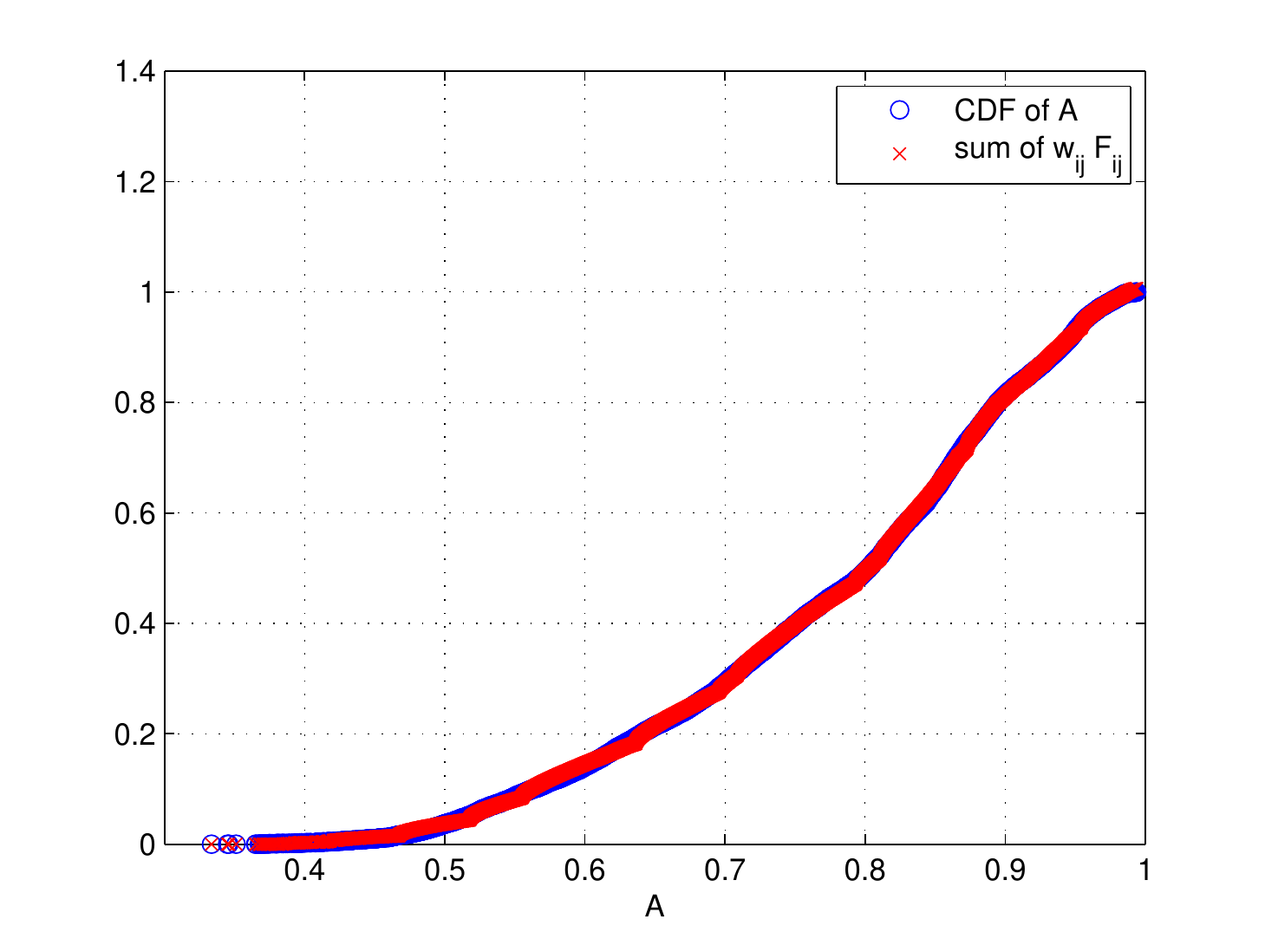} \caption{As Fig.\ \ref{fig:fitting01}, but here the occupation numbers $w_{ij}$ were solved using a population of $100$ asteroids. Compared to Fig,\ \ref{fig:fitting01}, the function series converges to the CDF more accurately due to the large shape population. The relative error is about 0.5\%. This shows that, for large populations, the shape deviations from ellipsoids actually cancel the systematic effects of each other, so the basis functions $F_{ij}$ are very good for describing the CDF.
} \label{fig:fitting02} \end{figure}

\section{Simulations} \label{chap:simu}

We have roughly three major sources of error in our simulations. First, assume Gaussian noise on the brightness function, $\varepsilon_L \sim \mathrm{N}(0, \, \sigma^2)$. We can assume the relative noise level is fairly small, $\vert{ \varepsilon_L }\vert \lesssim 0.03 L$. Second, while the observations for different values of $\theta$ are numerous (thousands of targets are observed, but most of them only at one $\theta$, so individual targets cannot be modelled), the observations of $\phi$ are typically scarce in most sky surveys to which our methods pertain. For our simulations, we have a minimum sample of 10 observations of $\phi$ for each asteroid if $\eta$ is to be estimated.  Third, our model is incorrect, as the actual shape of an asteroid is actually not an ellipsoid. 

We observed that the systematic error of the model dominates over the noise on the brightness function. Assuming a reasonably bounded noise level, the effect of the noise can be considered insignificant. The limited number of $\phi$ causes inaccuracy in the brightness function, which makes the estimation of $\eta$ harder, especially with a small sample of asteroids. However, if we have a small sample population, the model error still dominates over the low number of observations. For a sample of $\gtrsim 20$ asteroids and 10 observations of $\phi$ for each brightness measurement, the error in $\eta$ becomes small, and apparently the incorrect model is practically the only cause of error.

Fig.\ \ref{fig:fitting01} illustrates the effect of the model shape error.
We have plotted the function series $\sum_{i, \, j} w_{ij} F_{ij}(A)$ from Eq.\ \eqref{Fijeq} in the same plot with a simulated CDF $C(A)$. The occupation numbers $w_{ij}$ have been solved from Eq.\ \eqref{opt}; obviously the ellipsoidal basis functions do not converge fast if simulations are done with one non-ellipsoidal shape. However, as can be seen from Fig.\ \ref{fig:fitting02}, when the CDF is constructed from a population of varying non-ellipsoidal shapes, the basis functions $F_{ij}$ converge well.

Let us consider a sample of $S$ asteroids. For each asteroid, we measure their brightness function $L_l$, where $l = 1$, $\ldots$, $S$, and use them to compute their amplitudes $A_l$ at various geometries from the $\eta$-estimates. In our simulations, we have $S = 100$; such a population size is usually large enough to ensure the convergence of the basis functions while maintaining fast computation times.

First, we attempt to obtain information from only the $p$-distribution, while the $\beta$-distribution is assumed uniform. In our inverse problem, we choose a grid of 20 points for $p$, where $p_i \in [ \frac{i-1}{20}, \frac{i}{20} ]$, and $i = 1$, $\ldots$, $20$. For the $\eta$ estimation method using 10 observations of $\phi$ for each asteroid, we choose $\delta_p \propto 10^{-4}$. This amount of regularization is usually sufficient, as any more would smooth the solution too much and make the peaks of the $p$-distribution too wide. For the two-point scattering case, we choose $\delta_p \propto 10^{-3}$. The $p$-distribution of the forward model and the occupation numbers $w_i$ solved from the inverse problem using both $\eta$ and $q$ as data have been plotted in Fig.\ \ref{fig:p-only}. As we can see, it is possible to obtain accurate information of the $p$ distribution even with a low number of $\phi$ observations. For both methods of sections \ref{sec:eta} and \ref{sec:2point}, the location of the peak was recovered accurately. Of the two, the case of $\eta$-data is more accurate, as expected.

\begin{figure}[!ht] \centering \includegraphics[width=0.32\textwidth]{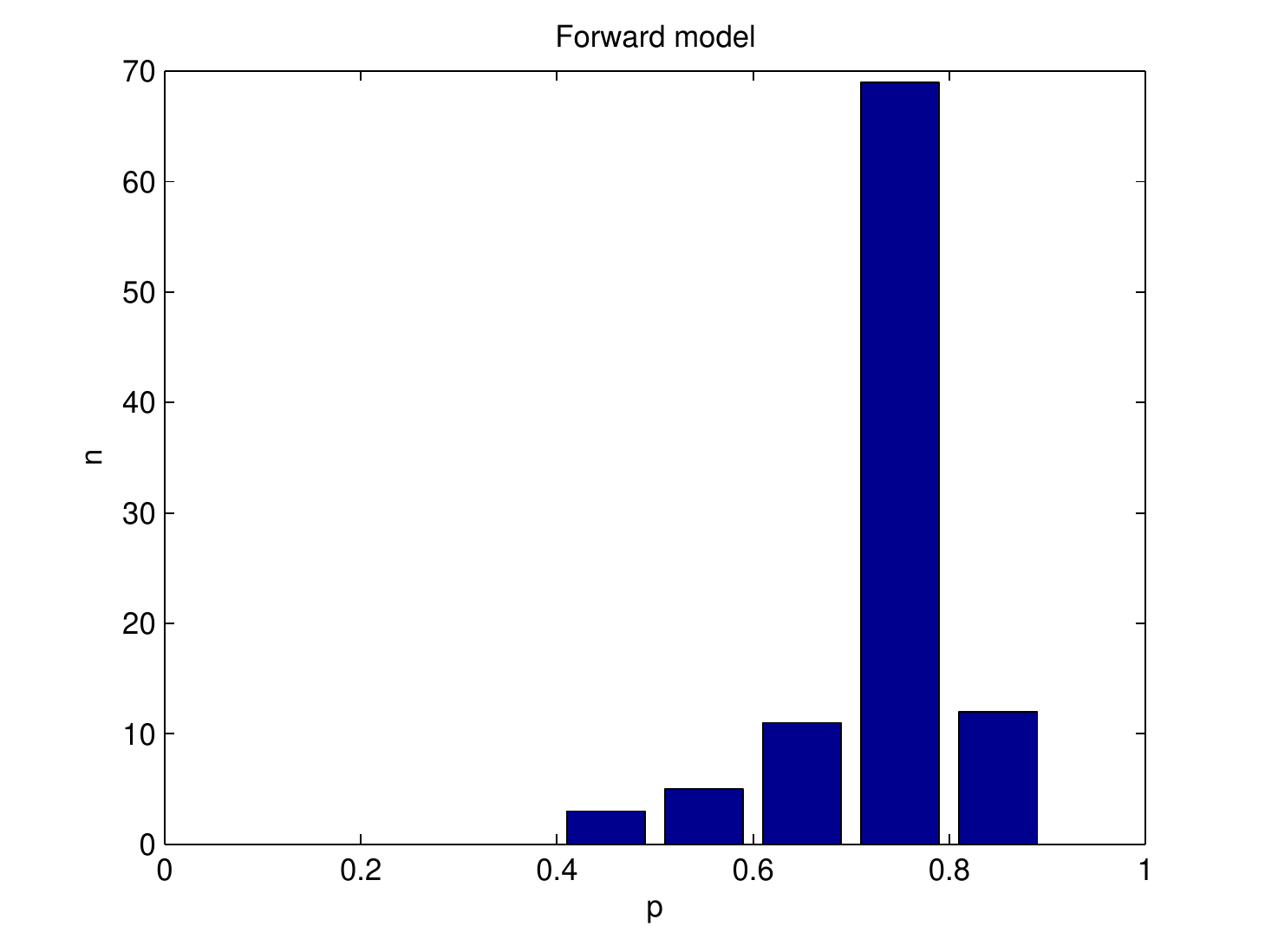} \includegraphics[width=0.32\textwidth]{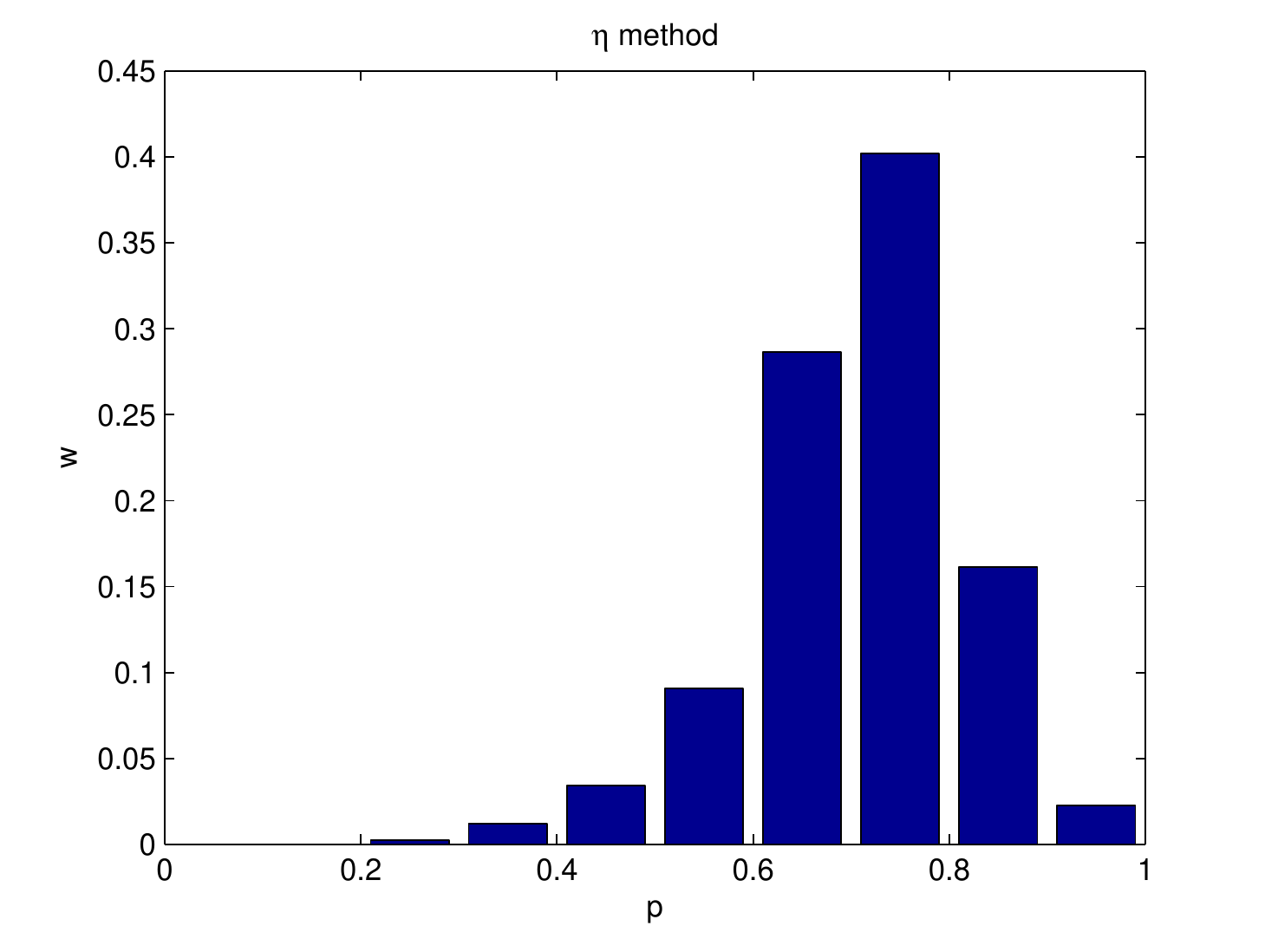} \includegraphics[width=0.32\textwidth]{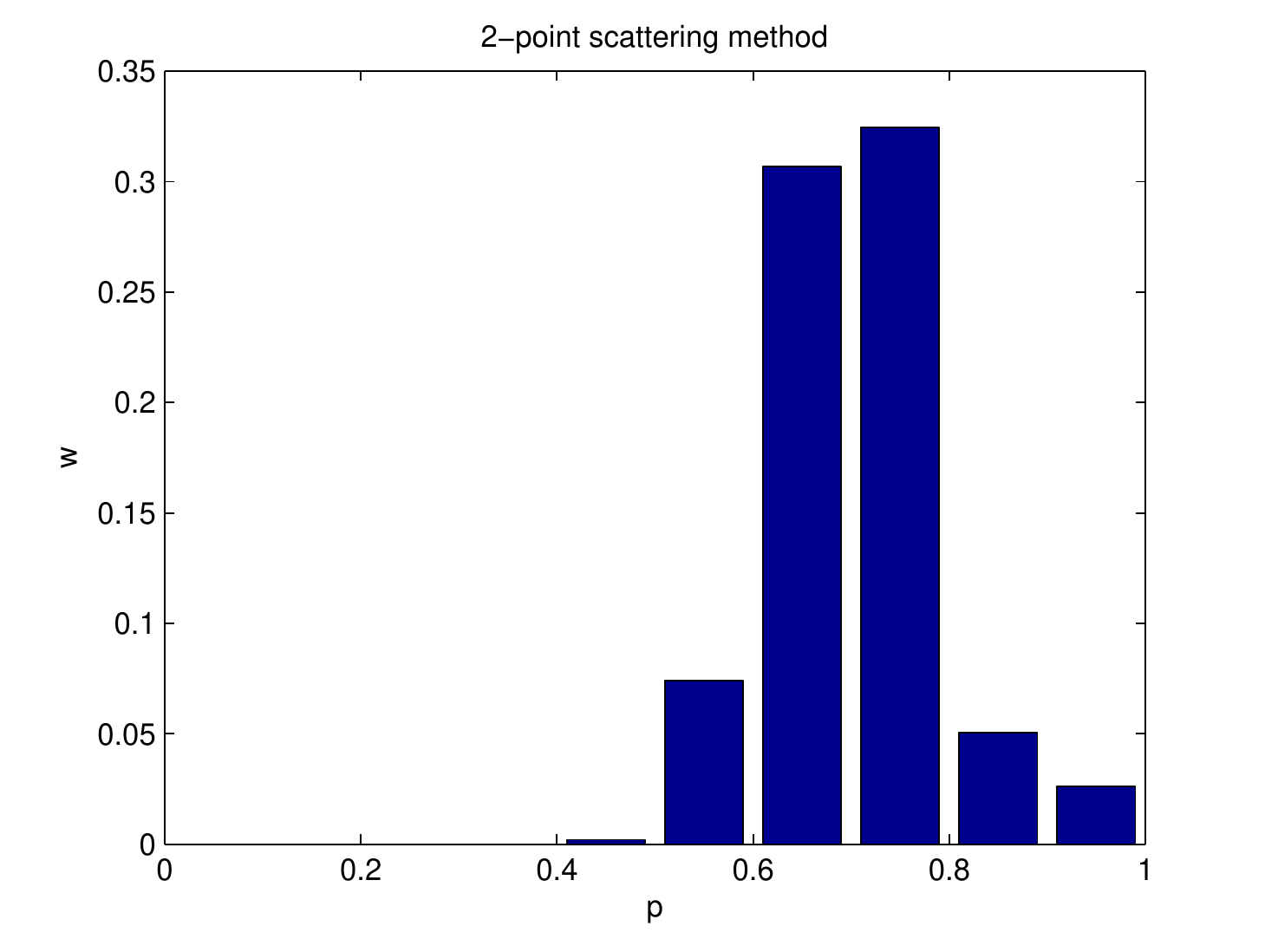} \caption{The $p$-distribution of the forward model (left), and the occupation numbers solved from the inverse problem for the $\eta$ estimation method (middle) and the two-point scattering method (right).} \label{fig:p-only} \end{figure}

We now move on to obtaining the joint $(p, \beta)$-distribution. The simulations were carried out by placing asteroids with a given spin $\beta$-distribution in orbits in the inertial $xy$-plane (the plane of the Earth's orbit), and by using uniformly distributed positions in the orbits for observation epochs. This creates data that contain information on $\beta$ that can be extracted with the $F_{ij}$-functions. We note that our simulation automatically avoided "inverse crime" since the synthetic data were constructed with a model entirely different from the one used in inversion. On the other hand, because of this, the recovered distributions can never be expected to be very accurate.

As the $\beta$-distribution is considerably harder to obtain accurately than the $p$-distribution, more regularization is needed for $\beta$. However, as we are more interested in locating the peaks of $p$- and $\beta$-distributions than obtaining a perfectly smooth solution, we choose fairly small values for our regularization parameters. In our simulations, we choose $\delta_{p} \propto 10^{-1}$ and $\delta_{\beta} \propto 1$. These values have been chosen using a combination of iteration and the standard L-curve method. As for our grid, we choose 20 different values for $p$ and 19 different values for $\beta$, where $p_i \in [ \frac{i-1}{20}, \frac{i}{20} ]$ and $\beta_j \in [\frac{j-1}{19} \frac{\pi}{2}, \frac{j}{19} \frac{\pi}{2} ]$, $i=1$, $\ldots$, $20$ and $j=1$, $\ldots$, $19$. Hence, we have $20 \cdot 19$ bins of $(p_i, \beta_j)$ in total. Using a population of 100 asteroids, the function series $\sum_{i, \, j} w_{ij} F_{ij}(A)$ from Eq.\ \eqref{Fijeq} has been plotted in the same plot with the CDF $C(A)$ in Fig.\ \ref{fig:fitting02}. With a large population, the function series with the occupation numbers solved from the inverse problem creates an excellent fit to the CDF. 

In our simulations, we tested mainly the resolution level of the data rather than tried to create physically realistic cases (which are smoother than those simulated here). We first attempted to recover information on a simple DF $f(p, \beta)$ with only one peak. A contour plot of the forward model and the $(p, \beta)$-distribution obtained from solving the occupation numbers from the inverse problem are shown in Fig.\ \ref{fig:contour01}. The solution is well acceptable as the location of the peak has been recovered in both $p$- and $\beta$-directions, despite some inevitable blurring.

To study the accuracy of the method in finding DF details, we consider DFs with bimodal distributions in $p$ and $\beta$. 
A contour plot of the actual joint distribution and the obtained solution are shown in Fig.\ \ref{fig:contour02}.
As could be expected, the solution for the parameter $p$ is more robust out of the two parameters. It is possible to obtain the distribution for $\beta$ with moderate accuracy, and even if the actual positions of the observed peaks have been slightly shifted, the existence of the peaks is confirmed in the obtained solution.

We experimented with the two-point scattering data as well, but we found that, due to the effectively linear dependence of the basis functions, virtually no information on the joint $(p, \beta)$-distribution could be obtained. Adding prior information of the joint distribution did not help, either, as it resulted in too heavy regularization, causing the solution to become almost entirely prior-based. Therefore, we conclude that, in the case of two-point scattering, the minimal number of $\phi$ angles and the model error cause the problem to be too unstable, preventing the possibility to obtain $\beta$-information.

\begin{figure}[!ht] \centering \includegraphics[width=0.495\textwidth]{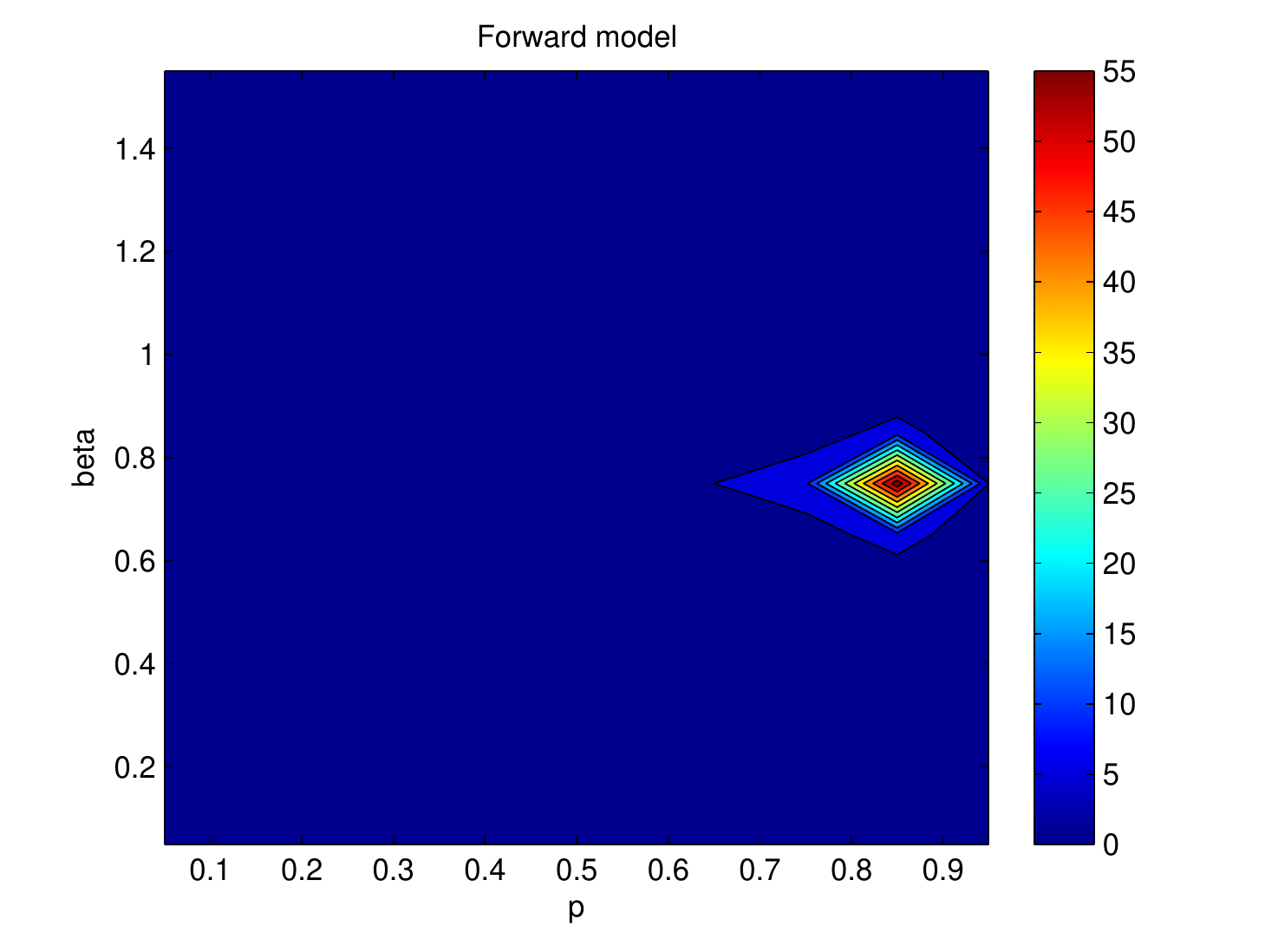} \centering \includegraphics[width=0.495\textwidth]{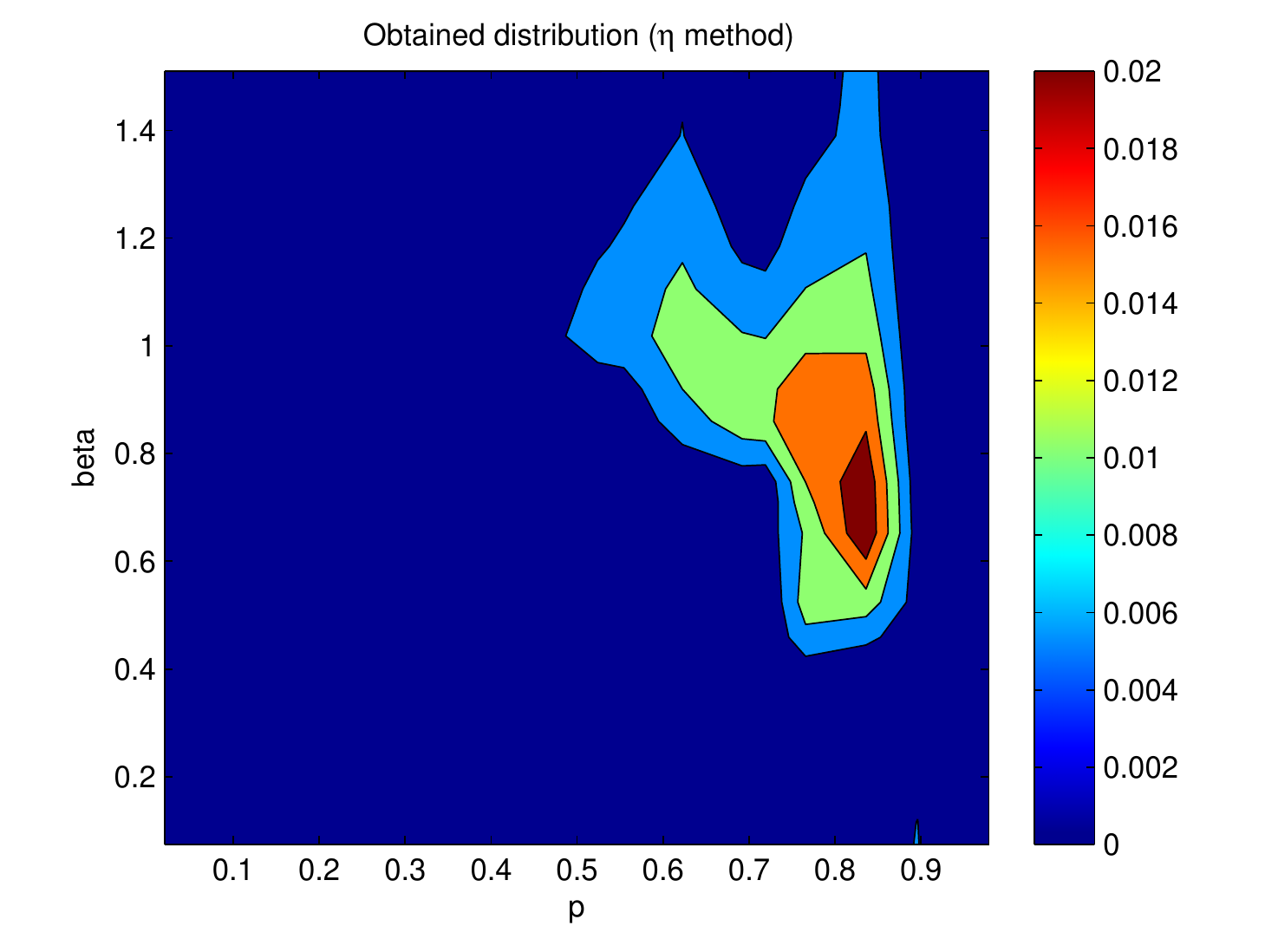} \caption{A contour plot of the actual $(p, \beta)$-distribution (left), and the solution of the inverse problem (right). The ``height'' of the contour represents the occupation numbers of the parameters $p$ and $\beta$, presented as weights on the right.} \label{fig:contour01} \end{figure}


\begin{figure}[!ht] \centering \includegraphics[width=0.495\textwidth]{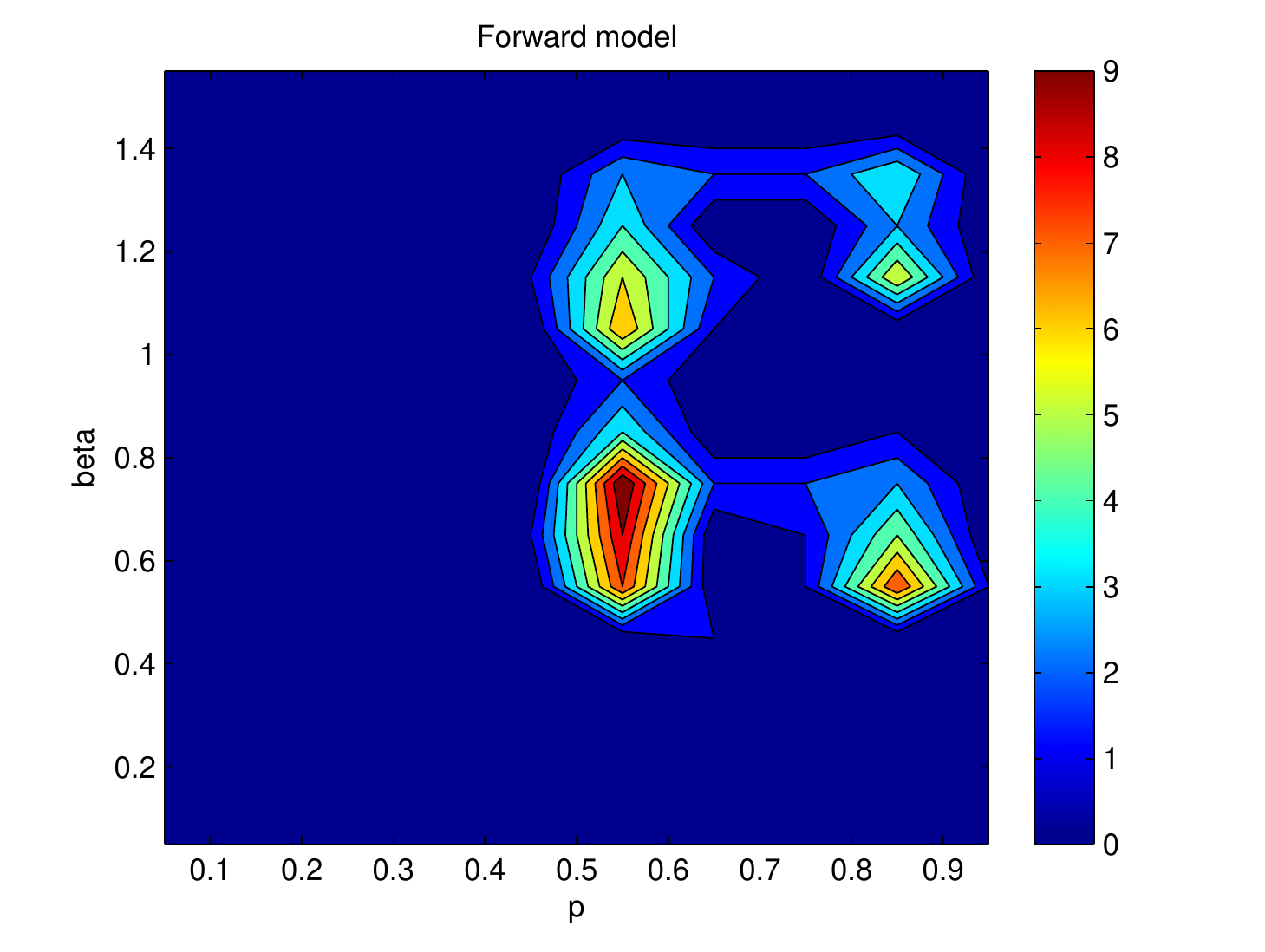} \centering \includegraphics[width=0.495\textwidth]{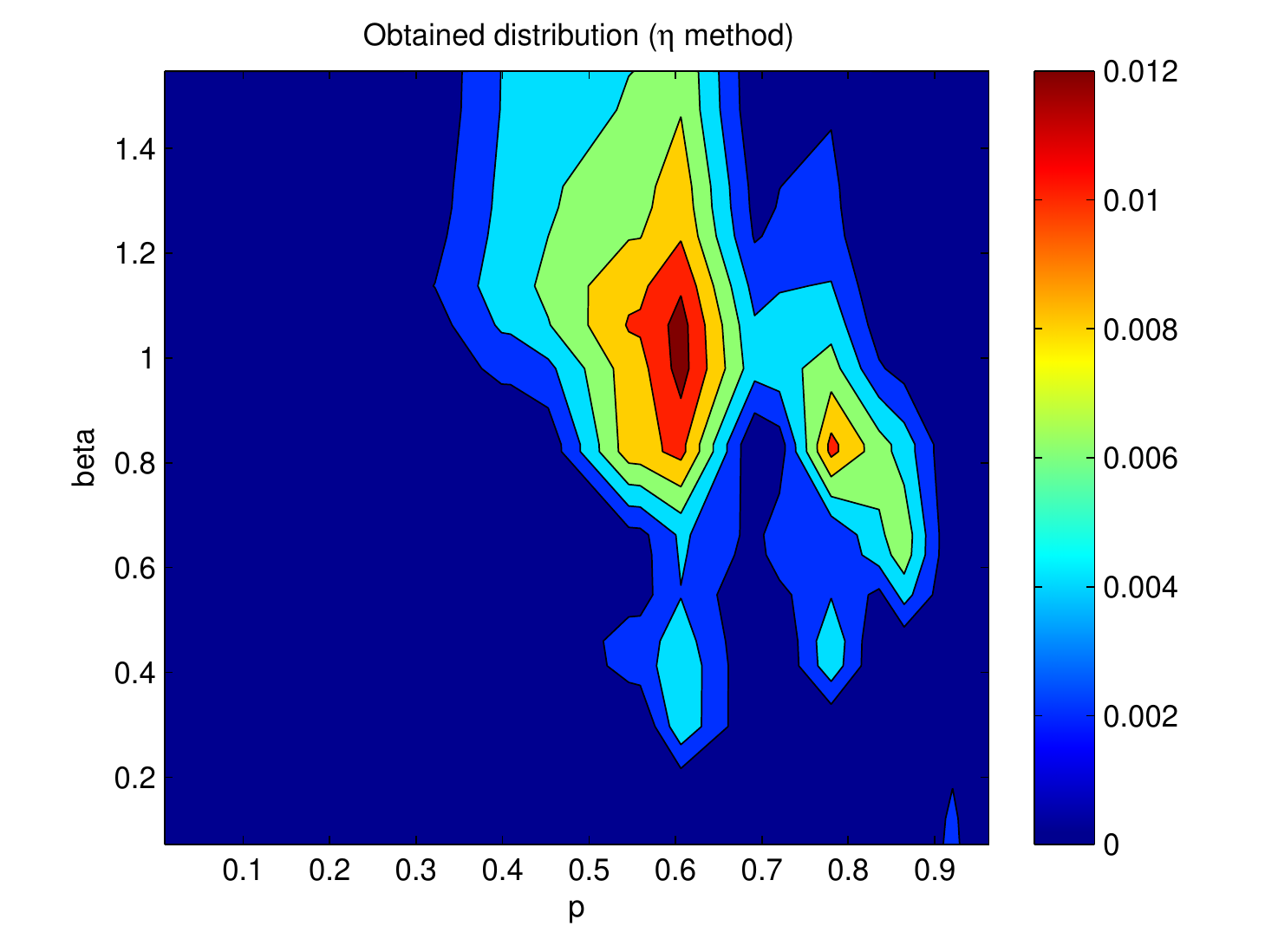} \caption{A multi-peaked contour plot of the actual $(p, \beta)$-distribution (left), and the solution of the inverse problem (right). } \label{fig:contour02} \end{figure}

\section{Conclusions and discussion}

We find that the simple ellipsoid model is, perhaps somewhat surprisingly, viable for estimating the basic shape and spin characteristics of object (asteroid) populations when the size of the population is sufficiently large (at least of order hundreds of targets or observations at various aspect angles $\theta$). We have presented a theoretical study of the shape and spin parameters and functions as well as a number of corresponding observables, and proven the uniqueness and stability properties of the inverse problem. The theoretical analysis was supported by simulations that confirmed the feasibility of our approach, and also showed that it is possible to obtain the distribution of the spin latitude parameter $\beta$ when there are enough observations per target, although this is less accurate than for the shape elongation $p$. If there are only two observations per target, information on $\beta$ is essentially lost. We have also given closed-form expressions of the basis functions of the two-point case derived in \cite{Szabo} by random sampling.

One can devise other observables in addition to those presented here. The main limitation is that, in practice, the measurements should consist of measures of brightness variation for single targets at fixed $\theta$ to remain invariant of scaling and other modelling aspects. One possibility is the derivative of the brightness variation (two measurements within a short time interval). However, in the inverse problem the derivative is a nonunique mixture of shape elongation and rotation rate (and rotation axis latitude), so the result is almost entirely dependent on what one chooses for prior distributions. For example, one could fix the elongation to one value and explain any observed derivative distribution by rotation rate distribution only.

The physical realization of the observables occurs mostly in large sky surveys and long observation campaigns. In these, one can often model individual targets even with of order one hundred data points as long as they cover a long time interval \cite{sparse}, but for many targets one only obtains a few data points that can be analyzed with the methods discussed here. In a study to be published elsewhere, we apply our methods to various real datasets of sufficient size for distribution analysis, and investigate the related practical issues by realistic simulations and tests.

\subsection*{Acknowledgements}
We thank Josef \v{D}urech and Riwan Kherouf for valuable discussions. This work was supported by the Academy of Finland (the centre of excellence in inverse problems).

\section*{Appendix}
\setcounter{section}{0}
\begin{lemma}
 If $f(p)$ and $C(A)$ are given as polynomials, $C(A)$ starts at the third degree, the polynomial coefficients of $f(p)$ are uniquely derivable from those of $C(A)$, and the attenuation factor of the polynomial coefficients of $f(p)$, when mapped to those of $C(A)$, is essentially inversely proportional to their degree.
\end{lemma}

\begin{proof}
Let us expand $f(p)$ as
$$
f(p)=\sum_{n=1}^\infty c_n p^n;\quad p\in[0,1].
$$
For isotropic $\theta$,
$$
C(A)=\sum_n c_n \int_0^A p^n\frac{\sqrt{A^2-p^2}}{\sqrt{1-p^2}}\, dp,
$$
and from tables of integrals we find that this is
$$
C(A)=A^2\sum_n c_n A^n \frac{1}{n+1} F_1(\frac{n+1}{2};\frac{1}{2},-\frac{1}{2};\frac{n+1}{2}+1;A^2,1),
$$
where $F_1$ is the Appell hypergeometric function. This form can be transformed into the usual Gauss hypergeometric function $_2F_1$ so that
$$
C(A)=A^2\sum_n c_n A^n k_n G_n(A),
$$
where
$$
G_n(x)=\ _2F_1(\frac{n+1}{2},-\frac{1}{2};\frac{n+4}{2};x^2)=\sum_j^\infty b_j^n x^{2j},
$$
with
$$
b_j^n=\frac{(\frac{n+1}{2})_j(-\frac{1}{2})_j}{j!(\frac{n+4}{2})_j}; \quad (a)_j=\frac{\Gamma(a+j)}{\Gamma(a)}
$$
(so $(a)_0=1=b_0^n$), and
$$
k_n=\frac{\sqrt{\pi}}{2(n+1)}\frac{\Gamma(\frac{n+3}{2})}{\Gamma(\frac{n+4}{2})},
$$
so $k_n\ne 0$ decreases monotonously as $n$ increases, and $\lim_{n\rightarrow\infty} k_n=0$. The decrease is moderate, approximated by, e.g., $\sim (n+1)^{-1}[\log (n/2+3)]^{-3/2}$ for $n<100$.
For the gamma function, $\Gamma(n+1/2)=\sqrt{\pi}(2n-1)!!/2^n$ and $\Gamma(n)=(n-1)!$.

Suppose the observed $C(A)$ is expanded (to hold for $0\le A\le 1$) as
$$
C(A)=A^2\sum_{n=1}^\infty a_n A^n.
$$
Then
$$
a_1=c_1k_1\Rightarrow c_1=a_1/k_1;\quad c_2=a_2/k_2;\quad a_3=c_3k_3+c_1k_1b_1^1\Rightarrow c_3=(a_3-c_1k_1b_1^1)/k_3,
$$
and so on recursively; i.e.,
$$
c_n=\frac{1}{k_n}(a_n-\sum_{i=1}^{[n]-1} c_{n-2i} k_{n-2i} b^{n-2i}_i),
$$
where $[n]$ is $(n+1)/2$ or $n/2$ for, respectively, odd or even $n$. 
\end{proof}

\begin{remark} There is a one-to-one mapping between the polynomial coefficients
determining $f(p)$ and $C(A)$.
If $p\in[0,1[$, and we expand (assuming $f(p)$ to vanish fast enough when $p\rightarrow 1$)
$$
\frac{f(p)}{\sqrt{1-p^2}}=\sum_{n=1}^\infty d_n p^n,
$$
we have
$$
C(A)=\sum_n d_n \int_0^A p^n \sqrt{A^2-p^2} \,dp=
A^2\sum_n d_n A^n\frac{_2F_1(-\frac{1}{2},\frac{n+1}{2};\frac{n+3}{2};1)}{n+1},
$$
which is simply
$$
C(A)=A^2\sum_n d_n k_n A^n,
$$
so we obtain
$$
d_n=\frac{a_n}{k_n}.
$$
\end{remark}

\begin{remark}
The uniqueness and stability results hold for the general triaxial ellipsoid as well.  
Let us now have a fixed $c\ne1$, $b=1$, and $a=1/p$. Then
$$
A^2=\frac{p^2\sin^2\theta+c^{-2}\cos^2\theta}{\sin^2\theta+c^{-2}\cos^2\theta},
$$
so the iso-$A$ curves are given by
$$
\cos^2\theta_{A3}(p):=g_{A3}(p)=\frac{A^2-p^2}{h(A)-p^2},
$$
where
$$
h(A):=A^2(1-c^{-2})+c^{-2}.
$$
Now
$$
C(A)=\sum_n c_n \int_0^A p^n\sqrt{g_{A3}(p)}\, dp,
$$
and this is
$$
C(A)=\frac{A^2}{\sqrt{h(A)}}\sum_n c_n A^n \frac{1}{n+1} F_1(\frac{n+1}{2};\frac{1}{2},-\frac{1}{2};\frac{n+1}{2}+1;\frac{A^2}{h(A)},1).
$$
This can be used to define a series expansion for the observed $C(A)$ with new basis functions instead of polynomials, so we have the same kind of one-to-one correspondence as above.
\end{remark}

\end{document}